\documentclass[10pt]{article}
\usepackage[cmex10]{amsmath}
\usepackage{amsthm}
\usepackage{enumerate}
\usepackage{subfig}
\usepackage{graphicx}
\usepackage{epstopdf}
\usepackage{amsfonts,amscd}
\usepackage{epsfig}
\usepackage{float}

\newtheorem{theorem}{Theorem}[section]

\newtheorem{conjecture}[theorem]{Conjecture}
\newtheorem{claim}[theorem]{Claim}
\newcommand{\be}{\begin{equation}}
\newcommand{\ee}{\end{equation}}
\newcommand{\bea}{\begin{eqnarray}}
\newcommand{\eea}{\end{eqnarray}}
\newcommand{\beann}{\begin{eqnarray*}}
\newcommand{\eeann}{\end{eqnarray*}}

\newcommand{\ket}[1]{\vert{#1}\rangle}
\newcommand{\bra}[1]{\langle{#1}\vert}
\newcommand{\unity}{{1\hskip -3pt \rm{I}}}

\begin{document}
\title{Classical capacity of a qubit depolarizing memory channel with {M}arkovian correlated noise}
\author{
    Jaideep Mulherkar\\
    Dhirubhai Ambani Institute of \\
    Information and Communication Technology\\
    jaideep\_mulherkar@daiict.ac.in
}
\maketitle
\begin{abstract}
We study the classical capacity of a forgetful quantum memory channel that switches between two qubit depolarizing channels according to an ergodic Markov chain. The capacity of this quantum memory channel depends on the parameters of the two depolarizing channels and the memory of the ergodic Markov chain. When the number of input qubit's is two, we show that depending on channel parameters either the maximally entangled input states or product input states achieve the classical capacity. Our conjecture based on numerics is that as the number of input qubits are increased the classical capacity approaches the product state capacity for all values of the parameters.
\end{abstract}
\section{Introduction}
\label{sec:Introduction}
Reliable transmission of classical information over quantum channels is an important problem in quantum information theory. The maximum amount of classical information that can be reliably transmitted over a quantum channel is called the classical capacity of the channel. Quantum memory channels model many physical situations where the noise effects are correlated. An unmodulated spin chain \cite{BBMB2008} and a micro maser \cite{GDF2009} have been proposed as physical models of quantum channels with memory effects. In recent years there has been a lot of interest in quantum channels with memory \cite{CGLM2014}. Capacity formulas for various classes of quantum memory channels were derived in \cite{KW2005}. As in the memoryless case, an important question is whether entangled inputs enhance the communication capacity of a quantum memory channel. For Pauli channels with memory, it was shown that below a certain threshold value of the noise correlation parameter, the maximally entangled states provide the optimal two-use classical capacity \cite{D2006}. An experimental demonstration of enhancement of classical information by entangling qubits for correlated Pauli channels was provided in \cite{BDWR2004}. In bosonic continuous variable memory channels also it is known that entangled inputs enhance capacity \cite{CCRM2006}.

A subclass of  quantum memory channels are forgetful channels where the effects of the initial memory configuration are forgotten over time. In this paper we look at a forgetful quantum memory channel where the noise correlations come from a classical ergodic Markov chain. Depending on the state of the Markov chain, one of the two depolarizing channels gets applied to the input qubit. This channel was first studied in \cite{WAFP2009} and the product state capacity was derived in terms of the entropy rate of a hidden Markov process . In the current work we investigate whether the classical capacity of this channel is enhanced if entangled inputs are allowed and give a partial answer to this question. We show that when the input consists of two qubits the channel capacity is achieved by either the maximally entangled input states or product input states depending on an explicit function of the parameters of the channel. Secondly, our numerical results point out that as the number of inputs qubits are increased the classical capacity approaches the product state capacity. In section \ref{sec:ChannelConstruction} we give a background on quantum memory channel, show the channel construction and review the result on the product state capacity of the channel. In section \ref{sec:EntEnhancement} show that for two qubits the classical information carrying capacity of maximally entangled states is better than the product states. Finally in section \ref{sec:ConjClassicalCap} we present numerical evidence to support that the product state capacity of this channel is equal to its classical capacity.

\section{Construction of the channel and its product state capacity}
\label{sec:ChannelConstruction}
A quantum channel is a completely positive trace preserving map $\Lambda:\mathcal{A}\rightarrow \mathcal{B}$ where $\mathcal{A}$ and $\mathcal{B}$ are the observable algebras of the input and output systems, respectively. Memoryless channels are channels where the noise acts independently on each input state. Multiple uses of a memoryless channel is given by the tensor product $\Lambda^{\otimes n}:\mathcal{A}^{\otimes n}\rightarrow \mathcal{B}^{\otimes n}$. The classical capacity of a quantum channel is the maximum rate at which classical information can be transmitted over the channel. The one shot classical capacity of a quantum channel is given by the Holevo capacity \cite{H2002}.
\beann
\chi^*(\Lambda) &=& \max_{\{p_i,\rho_i\}} S\big(\sum_i p_i\Lambda(\rho_i))- \sum_i p_i S(\Lambda(\rho_i)\big)  
\eeann
where $S(\rho) = -Tr(\rho \log(\rho))$ is the Von Neumann entropy of the state $\rho$ and the maximum is taken over all possible input ensembles. This is the capacity of the quantum channel when only product state encoding is allowed and is also known as the product state capacity. The n-use classical capacity by allowing entangled inputs is the amount of classical information that can be reliably transmitted per channel use is given by
\beann
C_n(\Lambda) = \frac{1}{n} \chi^*(\Lambda^{\otimes n})
\eeann
The classical capacity of the channel is given by 
\beann
C_{class}(\Lambda) = \lim_{n\rightarrow\infty} C_n
\eeann
Proving the additivity of the Holevo capacity was one of the most important problems in quantum information theory over the past decade. Additivity implies that entangled inputs cannot enhance the rate of information transmission. Additivity of the Holevo capacity has been shown to be true in depolarizing, unital and entanglement breaking channels \cite{K2010}. The superadditivity of the Holevo capacity in certain higher dimensional quantum channels was show by Hastings in \cite{H2009} thus disproving the additivity conjecture.

When a tensor product structure of the multi use of the channel does not hold we fall in the regime of channels with memory. Amongst the physically relevant models, we follow the non-anticipatory model of quantum memory channel given in \cite{CGLM2014}. In this model of quantum memory channels, besides the input and output, there is a third system $\mathcal{M}$ which  represents the state of the memory. The channel operates on the input state and the state of the memory resulting in an output state and a new state of the memory. Thus a quantum memory channel is represented by a CPTP map $\Gamma:\mathcal{M}\otimes \mathcal{A}\rightarrow \mathcal{B} \otimes \mathcal{M}$. The action of n successive uses of this channel $\Gamma$ is  given by the channel $\Gamma_n:\mathcal{A}^{\otimes n} \otimes \mathcal{M}\rightarrow \mathcal{M} \otimes \mathcal{B}^{\otimes n}$ where 
\beann
\Gamma_n = (\unity_{\mathcal{B}}^{\otimes n-1} \otimes \Gamma)\circ (\unity_{\mathcal{B}}^{\otimes n-2} \otimes \Gamma\otimes \unity_{\mathcal{A}})\circ\cdots\circ(\Gamma\otimes\unity_{\mathcal{A}}^{\otimes n-1})
\eeann
The final output state can be determined by performing a partial trace over the memory system. A quantum channel is called forgetful if the memory behavior does not depend on the initial memory configuration. That is, if for any input state $\rho^{(n)}$  and $\epsilon > 0$ there exists an $N(\epsilon)$ such that for all $n \geq N(\epsilon)$
\beann
D(\Gamma_n(\rho^{(n)} \otimes \omega),\Gamma_n(\rho^{(n)} \otimes \omega^{'})) &<& \epsilon
\eeann
for any pair of initial memory $\omega, \omega^{'}$, and $D$ is the trace distance
\beann
D(\rho_1,\rho_2) := \frac{1}{2}\|\rho_1-\rho_2\|_1
\eeann
A depolarizing channel $\Lambda:M^d\rightarrow M^d$ is a quantum channel that retains its state with probability $x^0$ and moves to the completely mixed state with probability $1-x^0$
\beann
\Lambda(\rho) = x^0\rho + (1-x^0)\frac{\unity}{d}
\eeann
This map is completely positive for $-\frac{1}{3} \leq x^0 \leq 1$. In this paper we look at a special case of quantum memory channels where the state of the memory transits between state `0' and state `1' according to an ergodic Markov chain with transition matrix $E$. If the memory state is `0' then a depolarizing channel $\Lambda_0$ is applied and if the memory state is `1' then a depolarizing channel $\Lambda_1$ is applied where
\beann
\Lambda_0(\rho) = x_0^0\rho + (1-x_0^0)\frac{\unity}{d}  \quad \text{and} \quad \Lambda_1(\rho) = x_1^0\rho + (1-x_1^0)\frac{\unity}{d}
\eeann
The action of this quantum memory channel a state $\rho$ can be written as
\beann
\Gamma(\mu \otimes \rho) =  \Lambda_0(\rho)\otimes (E\mu)_0 + \Lambda_1(\rho)\otimes (E\mu)_1
\eeann
Successive application of this channel  results in the n qubit channel
\bea
\Gamma_n = \rho_1 \otimes \rho_2 \otimes \cdots\otimes \rho_n  \mapsto \sum_{i_1,i_2,...,i_n}\gamma_{i_1} p_{i_1 i_2}...p_{i_{n-1} i_n} \Lambda_{i_1}(\rho_1)\otimes \cdots \otimes \Lambda_{i_n}(\rho_n) 
\eea
For forgetful memory channels it was shown \cite{KW2005,DD2009} that the classical capacity of forgetful quantum memory channels is given by the regularized Holevo capacity
\begin{equation}
\label{eq:qmccapacity}
\begin{aligned}
C_{class}(\Gamma) &= \lim_{n \rightarrow \infty} \frac{1}{n}\chi^*(\Gamma_n)\\
\chi^*(\Gamma_n) &=  \max_{\{p_i^{(n)},\rho_i^{(n)}\}} S\Big(\sum_i p_i^{(n)}\Gamma_n(\rho_i^{(n)})\Big)- \sum_i p_i^{(n)}S\Big(\Gamma_n(\rho_i^{(n)})\Big)
\end{aligned}
\end{equation}
It was shown in \cite{WAFP2009} that the product states of the form 
\bea
\rho^{(n)}_{\vec{i}} = \ket{i_1}\bra{i_1}\otimes\cdots \otimes \ket{i_n}\bra{i_n} 
\eea
with $\vec{i} = \{i_1,i_2,...,i_n\}$ minimize the output entropy and the product state capacity $C_{prod}(\Gamma)$ is given by the following theorem 
\begin{theorem}[\cite{WAFP2009}]
The product state capacity of the quantum memory channel $\Gamma$ is given by
\beann
C_{prod} = 1 - \lim_{n\rightarrow \infty} \frac{S(\Gamma_n(\rho^{(n)}_{\vec{i}}))}{n}
\eeann
\end{theorem}
Moreover, the output state has the following  eigen decomposition
\beann
\Gamma_n(\rho^{(n)}_{\vec{i}}) = \sum_{\vec{k}}  \lambda_{\vec{k} } \rho_{\vec{i}\oplus\vec{k}} 
\eeann
where the eigen values $\lambda_{\vec{k} }$ are given by
\bea
\label{eq:Chouteigenvalues1}
\lambda_{\vec{k} }= \gamma_{k_1} p_{k_1 k_2}...p_{k_{n-1} k_n} \alpha_{i_1}^{k_1}... \alpha_{i_n}^{k_n}
\eea
where
\beann 
\alpha_i^0 = \frac{1+ x_i^0}{2} , \, \alpha_i^1 = \frac{1- x_i^0}{2}
\eeann
From \ref{eq:Chouteigenvalues1} we see that the limit $\frac{S(\Gamma_n(\rho^{(n)}_{\vec{i}}))}{n}$ is equal to the entropy rate of a hidden Markov process with stationary measure $\lambda_{\vec{k} }$. The observed process is whether a bit was flipped and the hidden process is which of the two depolarizing channel was used. Efficient computation of the entropy rate of a hidden Markov process is a longstanding problem in classical information theory. If one of the two depolarizing channels is a perfect channel then the capacity of this quantum memory channel can be computed efficiently by the techniques of algebraic measures\cite{FNS1992} which was shown in \cite{MMN2012,M2014}.

\section{Entanglement enhanced communication with two inputs}
\label{sec:EntEnhancement}
In this section we focus on the question whether the capacity given by equation (\ref{eq:qmccapacity}) is enhanced by using entangled inputs when the input consists of two qubits. The channel we study comes under the category of Pauli channels.  A single qubit Pauli channel is given by the mapping
\beann
\rho\mapsto \sum_i q_i \sigma_i\rho\sigma_i
\eeann
where $\sum_i q_i =1$ and $\sigma_{i}'s$ are the Pauli matrices given by
\begin{align*}
\sigma_0 &=
\begin{pmatrix}
1&0\\0&1
\end{pmatrix} &
\sigma_1 &=
\begin{pmatrix}
0&1\\1&0
\end{pmatrix} \\
\sigma_2 &=
\begin{pmatrix}
0&-i\\i&0
\end{pmatrix} &
\sigma_3 &=
\begin{pmatrix}
1&0\\0&-1
\end{pmatrix}
\end{align*}
A memoryless two qubit Pauli channel consists of two independent uses of the single qubit channel. In a quantum memory channel the Pauli rotations are not independent. An example of a well studied two qubit quantum memory channel is the following
\beann
\rho\mapsto \sum_{ij}p_{ij}\sigma_i\otimes\sigma_j
\eeann
where
\beann
p_{ij} = (1-\mu)q_iq_j + \delta_{ij}\mu q_j 
\eeann
In this channel with probability $\mu$ the same Pauli rotation is applied to the two successive input and with probability $1-\mu$ the two rotations are independent. It was shown \cite{D2006} that there exists threshold value of the memory $\mu$ above which  maximally entangled input states maximize the two qubit capacity of this channel whereas below the threshold value the product states maximize the capacity.

In the quantum memory channel that we study in this paper the channel transition probabilities are Markov  and for simplicity we look at the special case with the transition matrix $E = \begin{pmatrix}\frac{(1+\mu)}{2}& \frac{(1-\mu)}{2}\\ \frac{(1-\mu)}{2}& \frac{(1+\mu)}{2} \end{pmatrix}$ where $\mu \in (-1,1)$. This quantum memory  channel can be parameterized by three parameters: the memory $\mu$ of the Markov chain, $a = x_0^0 + x_1^0$ and $d = x_0^0 - x_1^0$  and henceforth we refer to it as $\Gamma_2^{\mu,a,d}$. To find the 2-qubit capacity of the channel we look at the action of the channel on a general pure input state of two qubit's
\beann
\ket{\psi_{\theta,\phi}} &=& \frac{\cos\theta\ket{00}+ e^{i\phi}\sin\theta\ket{11}}{\sqrt{2}}\\
\rho_{\theta,\phi}        &=& \ket{\psi_{\theta,\phi}}\bra{\psi_{\theta,\phi}}
\eeann
The channel $\Gamma_2^{\mu,a,d}$ is a covariant channel \cite{H2002}, that is, we have the relation
\bea
\label{eq:covariantch1}
\Gamma_2^{\mu,a,d}\big((\sigma_i\otimes\sigma_j) \rho_{\theta,\phi} (\sigma_i\otimes\sigma_j)\big) = (\sigma_i\otimes\sigma_j)  \Gamma_2^{\mu,a,d}(\rho_{\theta,\phi} ) (\sigma_i\otimes\sigma_j)
\eea
We also have the identity
\bea
\label{eq:covariantch2}
\sum_{ij}\big((\sigma_i\otimes\sigma_j) \rho_{\theta,\phi}  (\sigma_i\otimes\sigma_j)\big) = \unity
\eea
Thus for equiprobable ensemble of states $\big((\sigma_i\otimes\sigma_j) \rho_{\theta,\phi}  (\sigma_i\otimes\sigma_j)\big)$ we can see that from the capacity formula (\ref{eq:qmccapacity}) the first term is equal to $\log 2$ due to equation (\ref{eq:covariantch1})and the second term due to equation (\ref{eq:covariantch1}) is equal to $S(\Gamma_2^{\mu,a,d}(\rho))$. Thus we get
that for a two qubit channel, 2-use capacity is given by
\bea
\label{eq:qmctwoqbitcapacity}
C_2(\Gamma_2^{\mu,a,d}) = 1 - \min_{\theta,\phi} S(\Gamma_2^{\mu,a,d}(\rho_{\theta,\phi}))
\eea
We show that entangled inputs indeed enhance the capacity of the channel $\Gamma_2^{\mu,a,d}$. In fact, similar to results of \cite{D2006}, depending on the parameters we either have the maximally entangled input states or the product input states that maximize capacity for the channel with two inputs. Define 
\bea
f(\Gamma_2^{\mu,a,d}) = |a^2+\mu d^2| -2|a|
\eea
Our main result is
\begin{theorem}
For any channel $\Gamma_2^{\mu,a,d}$, if $ f(\Gamma_2^{\mu,a,d}) \geq 0$, then the two qubit capacity of the quantum memory channel $C_2(\Gamma_2^{\mu,a,d})$ is achieved by the maximally entangled input states otherwise it is achieved at the product input state.
\end{theorem}
\begin{proof}
According to equation (\ref{eq:qmctwoqbitcapacity}) to find the capacity of the channel $\Gamma_2^{\mu,a,d}$ we need to look for input states that minimize the output entropy. The action of $\Gamma_2^{\mu,a,d}$ on the input pure states $\rho_{\theta,\phi}$ is given by
\beann
\Gamma_2^{\mu,a,d}(\rho_{\theta,\phi}) = \begin{pmatrix} \alpha & 0 & 0 & \delta\\ 0 & \gamma & 0 & 0\\ 0 & 0 & \gamma & 0\\ \delta^* & 0 & 0 & \beta \end{pmatrix}
\eeann
The diagonal elements $\alpha ,\beta$ and $\gamma$ are given by
\begin{equation}
\label{eq:eigrelation}
\begin{aligned}
\alpha &= \cos^2\theta \lambda_{0,0} + \sin^2\theta \lambda_{1,1}\\
\beta  &= \sin^2\theta \lambda_{0,0} + \cos^2\theta \lambda_{1,1}\\
\gamma &= \lambda_{0,1} = \lambda_{1,0}
\end{aligned}
\end{equation}
where $\lambda_{0,0}$, $\lambda_{0,1}$, $\lambda_{1,0}$ and $\lambda_{0,0}$ are given by equation (\ref{eq:Chouteigenvalues1}) applied to the case $n=2$. That is,
\begin{small}
\begin{equation}
\begin{aligned}
\lambda_{0,0} &= \frac{(1+\mu)}{4}\Big[\frac{(1+x_0^0)(1+x_0^0)}{4}+\frac{(1+x_1^0)(1+x_1^0)}{4}\Big] + \frac{(1-\mu)}{4}\Big[\frac{2(1+x_0^0)(1+x_1^0)}{4}\Big]\\
\lambda_{0,1} &= \lambda_{1,0}=  \frac{(1+\mu)}{4}\Big[\frac{(1+x_0^0)(1-x_0^0)}{4}+\frac{(1+x_1^0)(1-x_1^0)}{4}\Big] + \\&\qquad\qquad\qquad\qquad\qquad\qquad \frac{(1-\mu)}{4}\Big[\frac{(1+x_0^0)(1-x_1^0)}{4}+\frac{(1-x_0^0)(1+x_1^0)}{4}\Big]\\
\lambda_{1,1} &= \frac{(1+\mu)}{4}\Big[\frac{(1-x_0^0)(1-x_0^0)}{4}+\frac{(1-x_1^0)(1-x_1^0)}{4}\Big] + \frac{(1-\mu)}{4}\Big[\frac{2(1-x_0^0)(1-x_1^0)}{4}\Big]
\end{aligned}
\end{equation}
\end{small}
To determine the off diagonal element $\delta$ we consider the action of the tensor product of depolarizing channels $\Lambda_a$ and $\Lambda_b$
\beann
\Lambda_a\otimes\Lambda_b(A) = x_a^0x_b^0 A + x_a^0(1-x_b^0)\text{Tr}_1(A) + (1-x_a^0)x_b^0 \text{Tr}_2(A) + \frac{\unity\otimes\unity}{4}\text{Tr}(A)
\eeann
Therefore
\beann
\Lambda_a\otimes\Lambda_b(\ket{00}\bra{11}) = x_a^0x_b^0 \ket{00}\bra{11}
\eeann
From this we get
\bea
\delta = ce^{i\phi}\cos\theta\sin\theta = \frac{ce^{i\phi}\sin2\theta}{2}
\eea
where
\beann
c = \frac{1}{4}\big((1+\mu)(x_0^0)^2 + (1+\mu)(x_1^0)^2 +2(1-\mu)x_0^0x_1^0\big)
\eeann
The four eigenvalues of the output state counting multiplicities are 
\bea
\label{eq:Chouteigenvalues2}
\Big\{\frac{(\alpha + \beta) + \sqrt{(\alpha-\beta)^2 + 4|\delta|^2}}{2},\gamma,\gamma ,\frac{(\alpha + \beta) - \sqrt{(\alpha-\beta)^2 + 4|\delta|^2}}{2}\Big\}
\eea
with $\alpha, \beta,\gamma$ and $\delta$ given by equation (\ref{eq:eigrelation}). From equation (\ref{eq:eigrelation})we also have
\beann
\alpha + \beta &=& \lambda_{0,0} + \lambda_{1,1}\\
\alpha - \beta &=& \cos2\theta(\lambda_{0,0} -\lambda_{1,1})\\
\eeann
and therefore
\beann
\sqrt{(\alpha-\beta)^2 + 4|\delta|^2} = \sqrt{(\lambda_{0,0} -\lambda_{1,1})^2\cos^2 2\theta + c^2\sin^2 2\theta }
\eeann
If $(\lambda_{0,0} -\lambda_{1,1})^2 \geq c^2$ then $\sqrt{(\alpha-\beta)^2 + 4|\delta|^2}$  is maximized when $\theta =0$ else if $(\lambda_{0,0} -\lambda_{1,1})^2 \leq c^2$ then it is maximized when $\theta = \frac{\pi}{4}$. Correspondingly if $(\lambda_{0,0} -\lambda_{1,1})^2 \geq c^2$ the output entropy is minimized (and hence channel capacity is maximized) when $\theta = 0$ (product states) and if $(\lambda_{0,0} -\lambda_{1,1})^2 \leq c^2$ then output entropy is minimized (and hence channel capacity is maximized) when $\theta = \frac{\pi}{4}$(maximally entangled states). A simple calculation gives that the condition $(\lambda_{0,0} -\lambda_{1,1})^2 \leq c^2$ is the same as $f(\Gamma_2^{\mu,a,d}) \geq 0$. 
\end{proof}
The output eigenvalues that correspond to the minimum output entropy when $f(\Gamma_2^{\mu,a,d}) \leq 0$  are obtained by substituting $\theta = 0$ in equation (\ref{eq:Chouteigenvalues2}) and as expected we get the eigenvalues to be $\{\lambda_{0,0},\lambda_{0,1},\lambda_{1,0},\lambda_{1,1}\}$.
The two qubit capacity if $f(\Gamma_2^{\mu,a,d}) \leq 0$ (product state inputs) is 
\bea
C_2(\Gamma_2^{\mu,a,d}) = 1 + \frac{\lambda_{0,0}\log\lambda_{0,0} +2\lambda_{0,1}\log2 \lambda_{0,1}+ \lambda_{1,1}\log\lambda_{1,1}}{2}
\eea
A straightforward calculation shows that
\beann
\frac{(\alpha + \beta) - \sqrt{(\alpha-\beta)^2 + 4|\delta|^2}}{2} = \gamma = \lambda_{0,1}
\eeann
and hence when $f(\Gamma_2^{\mu,a,d}) \geq 0$  the output eigenvalues are
\beann
\{1-3\lambda_{0,1},\lambda_{0,1},\lambda_{0,1},\lambda_{0,1}\}
\eeann
The two qubit capacity if $f(\Gamma_2^{\mu,a,d}) \geq 0$ (maximally entangled inputs) is 
\bea
C_2(\Gamma_2^{\mu,a,d}) = 1 + \frac{(1-3\lambda_{0,1})\log(1-3\lambda_{0,1})+ 3\lambda_{0,1}\log\lambda_{0,1}}{2}
\eea
\begin{figure}[H]
\centering
{
\includegraphics[height = 6cm,width=9cm]{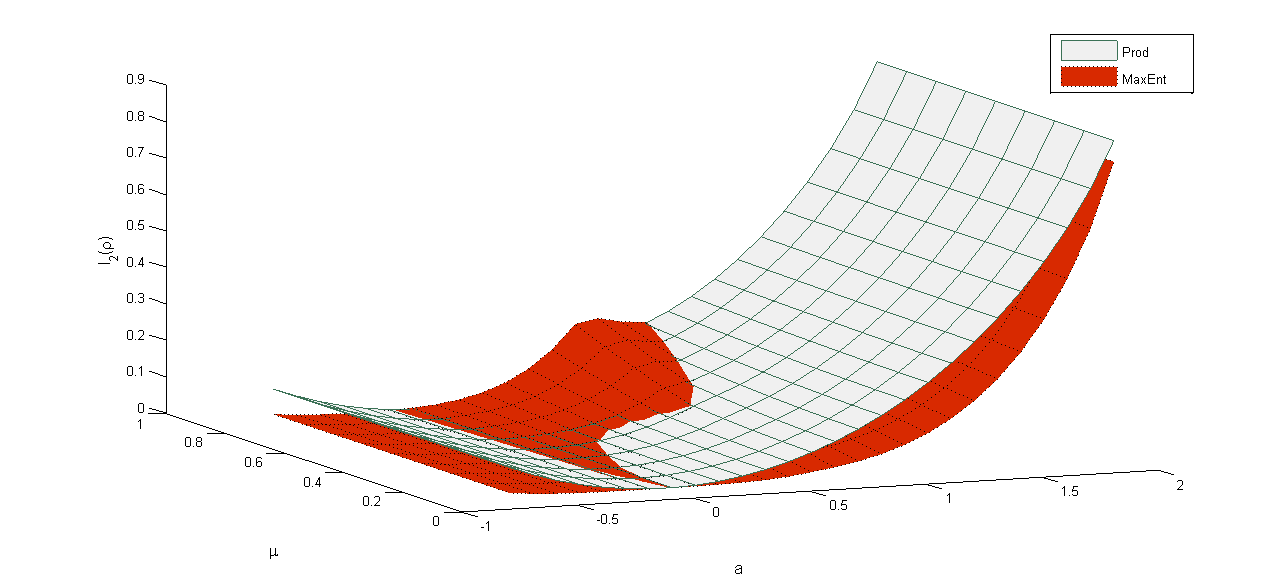}
\caption{The plot of two qubit mutual information against $a$ and $\mu$ for product and maximally entangled states. The parameter d is chosen to be the maximum value such that the depolarizing channel parameters are still in the range $[-1/3,1]$. }
\label{fig:qmc1}
}
\end{figure}

\begin{figure}[H]
\centering
{
\subfloat[][The plot of the two qubit mutual information $I_2(\rho)$ with change in the Markov memory  $\mu$ for the product states and the maximally entangled state. The other parameters are fixed to $a=\frac{1}{3}$ and $d = -1$. The two branches are corresponding to positive and negative values of $\mu$.]{
\includegraphics[width=5cm]{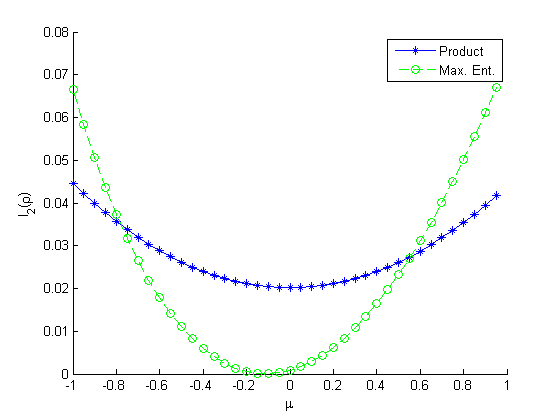}\qquad \includegraphics[width=5cm]{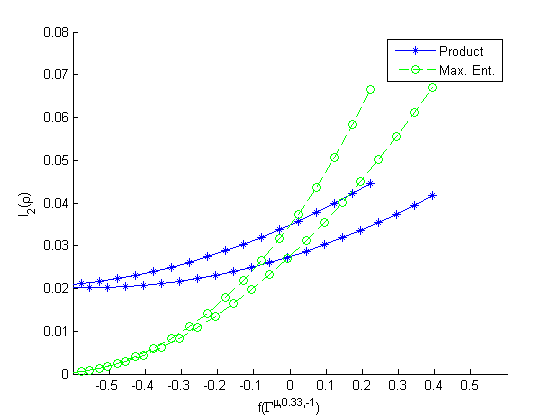}}
\newline
\subfloat[Subfigure 2 list of figures text][The plot of the two qubit mutual information $I_2(\rho)$ with change in the parameter $a=x_0^0+x_1^0$ for the product states and the maximally entangled state. The other parameters are fixed to $\mu=\frac{2}{3}$ and $d = -1$. ]{
\includegraphics[width=5cm]{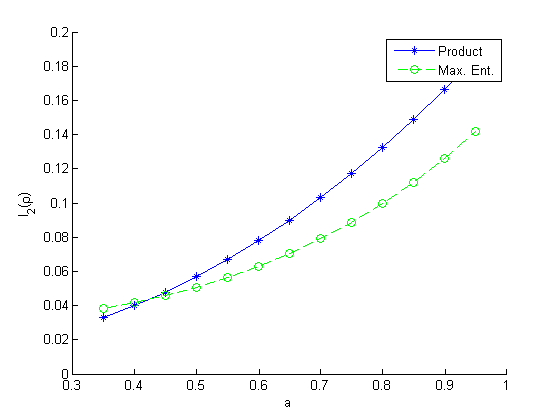}\qquad
\includegraphics[width=5cm]{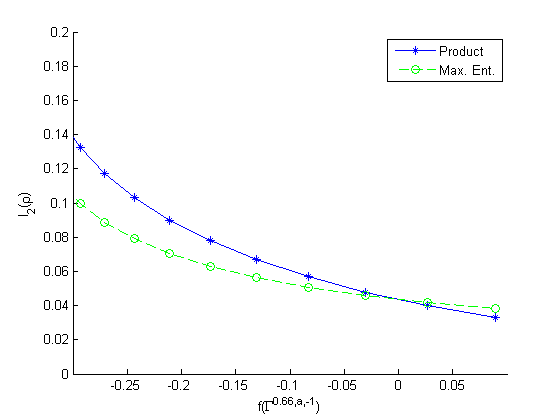}}
\newline
\subfloat[Subfigure 1 list of figures text][The plot of the two qubit mutual information $I_2(\rho)$ with change in the parameter $d=x_0^0-x_1^0$ for the product states and the maximally entangled state. The other parameters are fixed to $\mu=\frac{2}{3}$ and $a = \frac{1}{3}$.]{
\includegraphics[width=5cm]{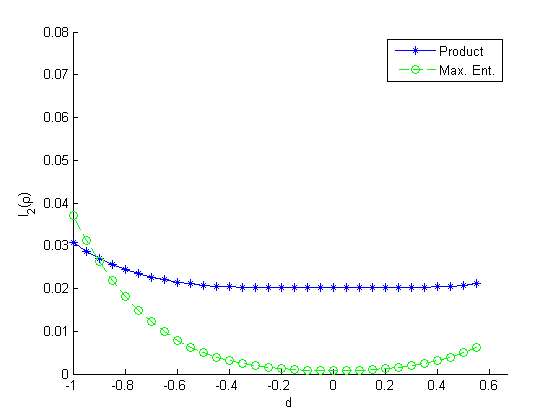}\qquad
\includegraphics[width=5cm]{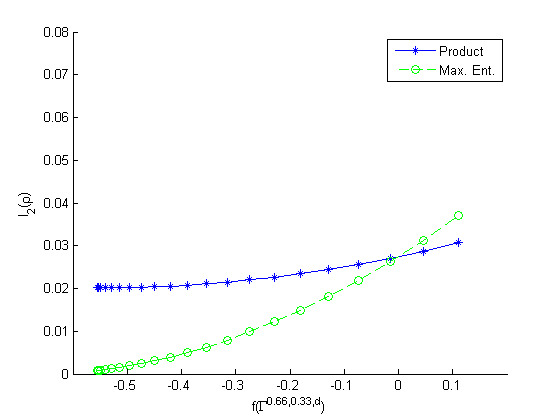}}
\caption{Plot of $I_2(\rho)$ for product and maximally entangled input states versus one of parameters $\mu,a,d$ by keeping the other two parameters fixed. The corresponding plot of $f(\Gamma_2^{\mu,a,d})$ is also shown. The capacity $C_2(\rho)$ is achieved either by the product states or the maximally entangled states. The crossover of the maximum of $I_2(\rho)$ from the iput states being maximally entangled to product states happens when  $f(\Gamma_2^{\mu,a,d}) =0.$ }
\label{fig:qmc2}
}
\end{figure}
Figures (\ref{fig:qmc1}) and (\ref{fig:qmc2}) shows the plot of the mutual information 
\beann
I_n(\rho) =S\Big(\sum_i p_i^{(n)}\Gamma_n(\rho_i^{(n)})\Big)- \sum_i p_i^{(n)}S\Big(\Gamma_n(\rho_i^{(n)})\Big)
\eeann
for the product states and maximally entangled states versus the parameters of the channel.
\section{Classical capacity of the channel}
\label{sec:ConjClassicalCap}
In this section we provide numerical evidence that the product state capacity of the quantum memory channel $\Gamma$ is equal to the classical capacity. For quantum memory channels that include a periodic channel with depolarizing branches and a convex combination of depolarizing channels the classical capacity was shown to be equal to the product state capacity \cite{MD2009}. We conjecture that for the quantum memory channel studied in this paper
\begin{conjecture}
\label{conj:Chconjecture}
\bea
C_{class}(\Gamma) = C_{prod}(\Gamma)
\eea 
\end{conjecture}
It is clear that $C_{class}(\Gamma) \geq  C_{prod}(\Gamma)$. Although we are unable to show the reverse inequality we give the following bound
\begin{theorem}
\beann
C_{class}(\Gamma) \leq C_{prod}(\Gamma) + \bar{H}(X)  
\eeann
where $\bar{H}(X)$ is the entropy rate of the Markov process with transition matrix $E$.
\end{theorem}
\begin{proof}
From the capacity formula of equation (\ref{eq:qmccapacity}) it is clear that
\bea
\label{eq:Capinequality}
C_n(\Gamma_n)  \leq 1 - \min_{\rho^{(n)}} \frac{S(\Gamma_n(\rho^{(n)}))}{n}
\eea
Define the quantities
\beann
s_{prod}(\rho_{prod}^{(n)})&:=&\sum_{i_1,...,i_n}\gamma_{i_1} p_{i_1 i_2}...p_{i_{n-1} i_n} S(\Lambda_{i_1}\otimes \cdots \otimes \Lambda_{i_n}(\rho_{prod}^{(n)}))\\
s(\rho^{(n)})&:=&\sum_{i_1,...,i_n}\gamma_{i_1} p_{i_1 i_2}...p_{i_{n-1} i_n} S(\Lambda_{i_1}\otimes \cdots \otimes \Lambda_{i_n}(\rho^{(n)}))\\
s_{prod} &:=& \min_{\rho^{(n)}_{prod}} s_{prod}(\rho_{prod}^{(n)})\\
s        &:=& \min_{\rho^{(n)}} s(\rho^{(n)})
\eeann
For an ensemble $\{p_i, \rho_i\}$ we have the inequality \cite{NC2000}
\beann
\sum_i p_i S( \rho_i) \leq S\big(\sum_i p_i \rho_i\big) \leq \sum_i p_i S( \rho_i) + H(p_i)
\eeann
where $H(p_i)=-\sum_i p_i \log p_i$. Using this inequality we get
\begin{equation}
\label{eq:Chinequality1}
\begin{aligned}
s_{prod}(\rho_{prod}^{(n)}) &\leq  &S(\Gamma_n(\rho_{prod}^{(n)}))&\leq  &s_{prod}(\rho_{prod}^{(n)}) &+ H(X)\\
s(\rho^{(n)})               &\leq  &S(\Gamma_n(\rho^{(n)}))       &\leq  &s(\rho^{(n)})               &+ H(X)
\end{aligned}
\end{equation}
where
\beann
H(X)&= -\sum_{i_1,...,i_n}(\gamma_{i_1} p_{i_1 i_2}...p_{i_{n-1} i_n})\log (\gamma_{i_1} p_{i_1 i_2}...p_{i_{n-1} i_n})
\eeann
Now we claim that 
\begin{claim}
\label{cl:qmcclaim}
\beann
s_{prod} \leq s
\eeann
\end{claim}
Due to claim (\ref{cl:qmcclaim}) and equation (\ref{eq:Chinequality1}) 
\beann
\min_{\rho^{(n)}_{prod}} S(\Gamma_n(\rho_{prod}^{(n)})) \leq s_{prod} + H(X) \leq s + H(X) \leq \min_{\rho^{(n)}} S(\Gamma_n(\rho^{(n)})) + H(X)
\eeann
and therefore,
\beann
\min_{\rho^{(n)}_{prod}} S(\Gamma_n(\rho_{prod}^{(n)}))  \leq \min_{\rho^{(n)}} S(\Gamma_n(\rho^{(n)})) + H(X)
\eeann
Dividing by $n$ and taking the limit as $n\rightarrow\infty$
\beann
\lim_{n\rightarrow\infty} \min_{\rho^{(n)}_{prod}} \frac{S(\Gamma_n(\rho_{prod}^{(n)}))}{n} \leq \lim_{n\rightarrow\infty} \min_{\rho^{(n)}} \frac{S(\Gamma_n(\rho^{(n)}))}{n} + \bar{H}(X)
\eeann
Substituting in equation (\ref{eq:Capinequality}) we get the required result.
\end{proof}
\textbf{Proof of Claim \ref{cl:qmcclaim}}
\begin{proof}
Firstly due to the additivity of the depolarizing channel \cite{K2010} we have that
\bea
\label{eq:Chinequality2}
\min_{\rho^{(n)}_{prod}} S(\Lambda_{i_1}\otimes \cdots \otimes \Lambda_{i_n}(\rho_{prod}^{(n)})) \leq \min_{\rho^{(n)}} S(\Lambda_{i_1}\otimes \cdots \otimes \Lambda_{i_n}(\rho^{(n)}))
\eea
We also have that
\begin{multline}
\label{eq:Chinequality3}
\sum_{i_1,...,i_n}\gamma_{i_1} p_{i_1 i_2}...p_{i_{n-1} i_n}\min_{\rho^{(n)}_{prod}} S(\Lambda_{i_1}\otimes \cdots \otimes \Lambda_{i_n}(\rho^{(n)_{prod}}))\leq\\ \min_{\rho^{(n)}_{prod}} \sum_{i_1,...,i_n}\gamma_{i_1} p_{i_1 i_2}...p_{i_{n-1} i_n} S(\Lambda_{i_1}\otimes \cdots \otimes \Lambda_{i_n}(\rho^{(n)}_{prod}))
\end{multline}
Also, for any compound channel that is a tensor product of depolarizing channels, the minimum output entropy is obtained on the product states of the form $\rho^{(n)}_{\vec{i}}$ given by equation (\ref{eq:Chouteigenvalues1}), thus states of the form $\rho^{(n)}_{\vec{i}}$ right and left side of the equation (\ref{eq:Chinequality3}) are equal and hence
\begin{multline}
\label{eq:Chequality}
\sum_{i_1,...,i_n}\gamma_{i_1} p_{i_1 i_2}...p_{i_{n-1} i_n}\min_{\rho^{(n)}_{prod}} S(\Lambda_{i_1}\otimes \cdots \otimes \Lambda_{i_n}(\rho^{(n)_{prod}})) =\\ \min_{\rho^{(n)}_{prod}} \sum_{i_1,...,i_n}\gamma_{i_1} p_{i_1 i_2}...p_{i_{n-1} i_n} S(\Lambda_{i_1}\otimes \cdots \otimes \Lambda_{i_n}(\rho^{(n)}_{prod}))
\end{multline}
Finally it is also clear that
\bea
\label{eq:Chinequality4}
s \leq \sum_{i_1,...,i_n}\gamma_{i_1} p_{i_1 i_2}...p_{i_{n-1} i_n}\min_{\rho^{(n)}} S(\Lambda_{i_1}\otimes \cdots \otimes \Lambda_{i_n}(\rho^{(n)}))
\eea
Combining equations (\ref{eq:Chinequality2}), (\ref{eq:Chequality}) and (\ref{eq:Chinequality4}) we get the required result.
\end{proof}
We provide figures (\ref{fig:qmc3}),(\ref{fig:qmc4}) and (\ref{fig:qmc5}) in support conjecture \ref{conj:Chconjecture}. In these we look at the mutual information $I_n(\rho)$ for a variety of entangled states including the W-state, GHZ-state and the maximally entangled state and compare it with the mutual information of the product state. For $n=4$ there is a small range of parameters (see figure) where the mutual information of the maximally entangled states is higher than that of the product states. By varying the parameters $\mu$, $a$ and $d$ we find that for $n=6$ and $n=8$ the product states have more mutual information than the entangled states.

\begin{figure}[H]
\centering
\includegraphics[width=10cm]{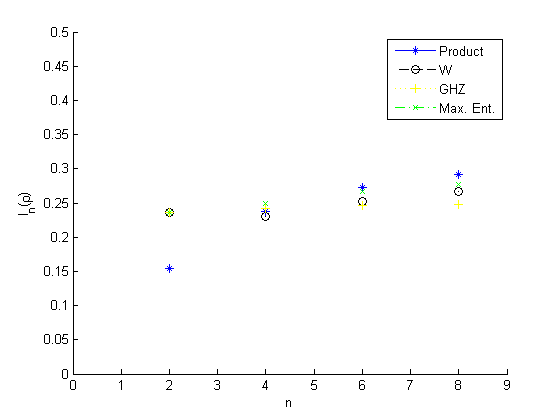}
\caption{Plot shows the mutual information $I_n(\rho)$ for $n=2,4,6,8$ for product states, GHZ-states, W-state and maximally entangled states for the channel $\Gamma_n^{\mu,a,d}$ for $\mu = 0.7$, $\mu=0.9,a = \frac{2}{3}$,$d=\frac{-4}{3}$. The mutual information of the product states exceeds that of the other states for $n=6$ and $n=8$.}
\label{fig:qmc3}
\end{figure}

\begin{figure}[H]
\centering
{
\includegraphics[width=6cm,height = 4cm]{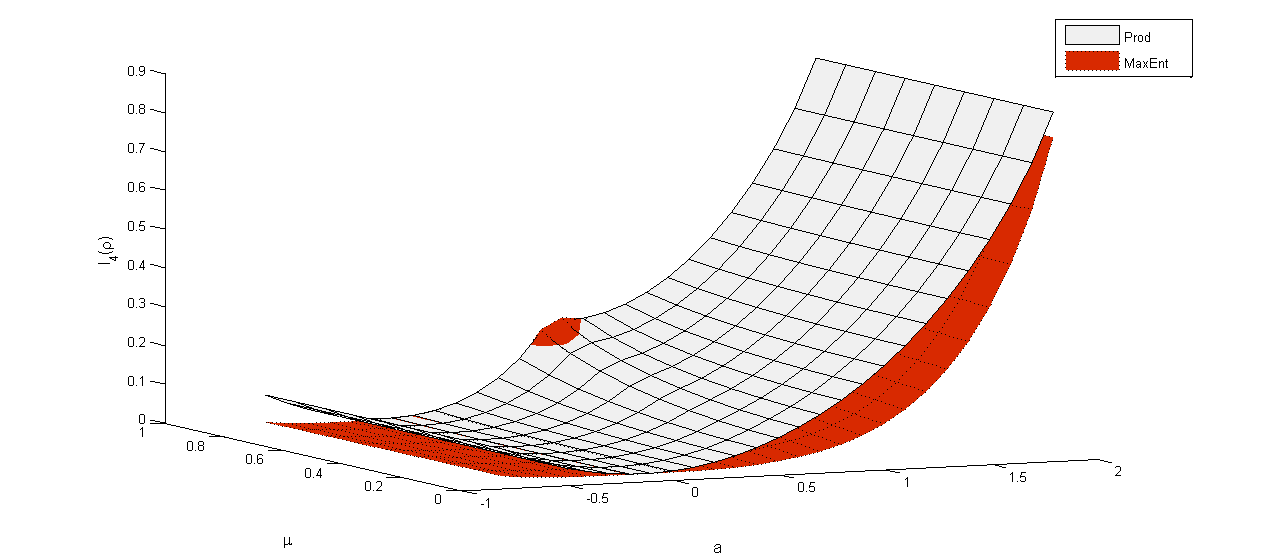} \qquad
\includegraphics[width=5cm]{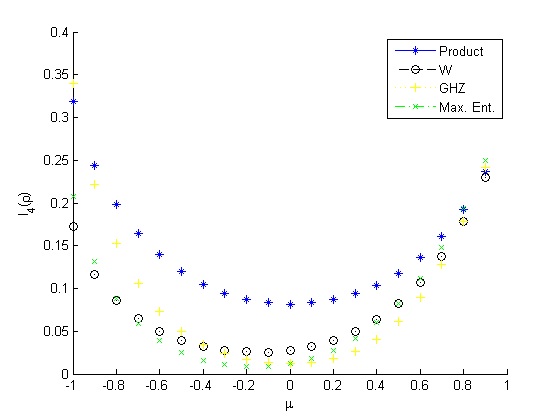}
}
\caption{On the left the plot of mutual information $I_4(\rho)$ against $a$ and $\mu$ for product and maximally entangled states. The parameter d is chosen to be the maximum value such that the depolarizing channel parameters are still in the range $(-1/3,1)$. On the right, plot of $I_4(\rho)$ versus $\mu$ ($a=\frac{1}{3}$, $d =\frac{4}{3}$) for product and different types entangled states including GHZ state, W-state and maximally entangled states.  In a small region of the parameter values the entangled states mutual information exceeds that of the product states.}
\label{fig:qmc4}
\end{figure}

\begin{figure}[H]
\centering
{
\includegraphics[width=6cm,height = 4cm]{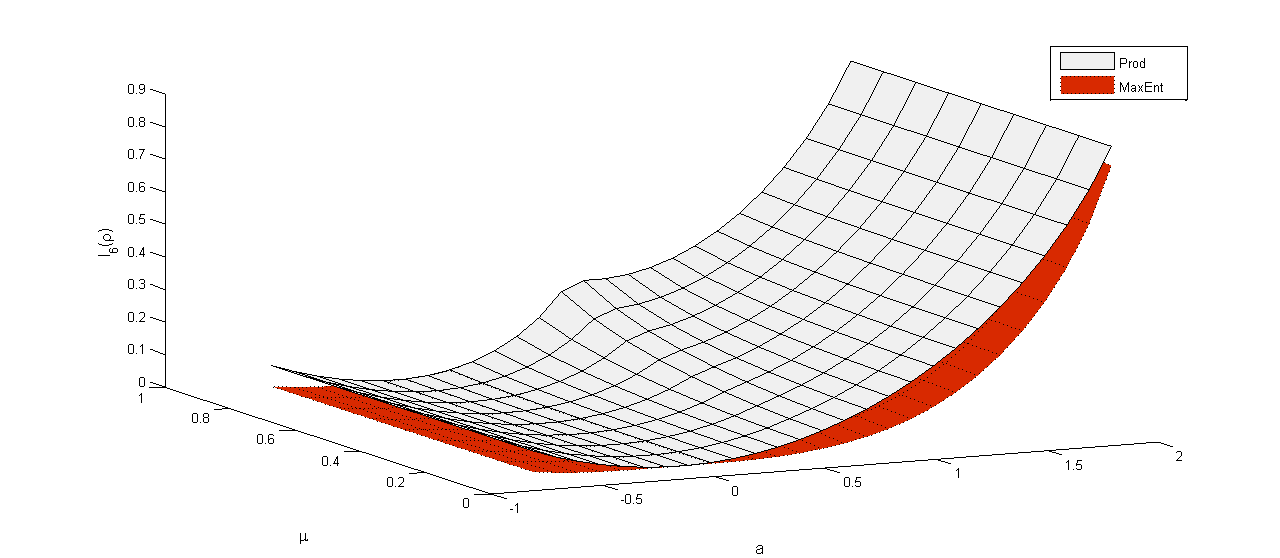} \qquad
\includegraphics[width=5cm]{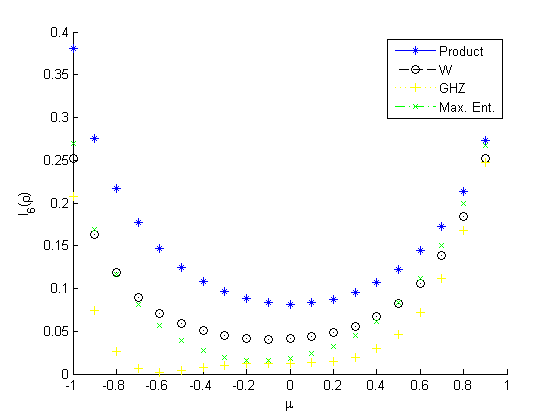}}
\caption{On the left the plot of mutual information $I_6(\rho)$ against $a$ and $\mu$ for product and maximally entangled states. For all the parameter values the product states mutual information exceeds that of the entangled states.}
\label{fig:qmc5}
\end{figure}

\bibliographystyle{ieeetr}
\bibliography{QMChannel}
\end{document}